\documentclass[11pt]{article}

\usepackage[utf8]{inputenc}
\usepackage[T1]{fontenc}
\usepackage{times}
\usepackage{latexsym}
\usepackage[margin=1in]{geometry}
\usepackage{amsmath,amssymb,amsthm}
\usepackage{mathtools}
\usepackage{booktabs}
\usepackage{multirow}
\usepackage{graphicx}
\usepackage{tikz}
\usetikzlibrary{arrows.meta,positioning,shapes.geometric,calc,
  decorations.pathreplacing,patterns,fit,backgrounds,
  matrix,chains,shadows}
\usepackage{pgfplots}
\pgfplotsset{compat=1.18}
\usepackage{algorithm}
\usepackage{algpseudocode}
\usepackage{tcolorbox}
\usepackage{enumitem}
\usepackage{xcolor}
\usepackage{hyperref}
\usepackage{cleveref}
\usepackage{subcaption}
\usepackage{float}
\usepackage{thmtools}
\usepackage{mathrsfs}

\definecolor{cblue}{HTML}{378ADD}
\definecolor{cpurple}{HTML}{7F77DD}
\definecolor{ccoral}{HTML}{D85A30}
\definecolor{cteal}{HTML}{1D9E75}
\definecolor{camber}{HTML}{BA7517}
\definecolor{cgreen}{HTML}{639922}
\definecolor{cpink}{HTML}{D4537E}
\definecolor{cgray}{HTML}{888780}
\definecolor{cred}{HTML}{E24B4A}
\definecolor{lblue}{HTML}{E6F1FB}
\definecolor{lpurple}{HTML}{EEEDFE}
\definecolor{lcoral}{HTML}{FAECE7}
\definecolor{lteal}{HTML}{E1F5EE}
\definecolor{lamber}{HTML}{FAEEDA}
\definecolor{lgreen}{HTML}{EAF3DE}

\newtheorem{theorem}{Theorem}[section]

\newtheorem{proposition}[theorem]{Proposition}

\theoremstyle{definition}
\newtheorem{definition}[theorem]{Definition}
\newtheorem{example}[theorem]{Example}
\newtheorem{remark}[theorem]{Remark}
\newtheorem{assumption}[theorem]{Assumption}

\newcommand{\E}{\mathbb{E}}
\newcommand{\R}{\mathbb{R}}

\newcommand{\Prob}{\mathbb{P}}
\newcommand{\KL}{\mathrm{D}_{\mathrm{KL}}}
\newcommand{\MI}{\mathrm{I}}
\newcommand{\Ent}{\mathrm{H}}

\newcommand{\Cobs}{\mathcal{C}_{\mathrm{obs}}}
\newcommand{\Ctype}{\mathcal{C}_{\mathrm{type}}}
\newcommand{\Cact}{\mathcal{C}_{\mathrm{act}}}
\newcommand{\Funct}{\mathcal{F}}
\newcommand{\GFunct}{\mathcal{G}}
\DeclareMathOperator*{\argmax}{arg\,max}

\DeclareMathOperator{\Var}{Var}

\DeclareMathOperator{\tr}{tr}

\DeclareMathOperator{\BR}{BR}

\DeclareMathOperator{\LTV}{LTV}

\title{Personalization as a Game: Equilibrium-Guided Generative Modeling for Physician Behavior in Pharmaceutical Engagement}

\author{
  Suyash Mishra \\
  AI Researcher \\
  \texttt{suyash.mishra@roche.com}
}

\date{}

\begin{document}
\maketitle

\begin{abstract}
We present \textbf{EGPF} (Equilibrium-Guided Personalization Framework), a mathematically rigorous architecture unifying Bayesian game theory, category theory, information theory, and generative AI for hyper-personalized physician engagement in the pharmaceutical domain. Our framework models the pharma--physician interaction as an incomplete-information Bayesian game where physician behavioral types are inferred via functorial mappings from observational categories, equilibrium strategies guide content generation through large language models (LLMs), and information-theoretic feedback loops ensure adaptive recalibration. We formalize behavior composition through category-theoretic functors, natural transformations, and monoidal structures, enabling modular, composable physician archetypes that respect structural invariants under domain shift. We introduce a novel \textit{Rate-Distortion Equilibrium} (RDE) criterion that bounds the personalization--privacy tradeoff, an \textit{Evolutionary Game Dynamics} layer for population-level behavior modeling, a \textit{Mechanism Design} module for incentive-compatible engagement, and a \textit{Sheaf-Theoretic} extension for multi-scale behavioral consistency. We prove convergence of our iterative belief-update mechanism at rate $O(\frac{K\log K}{t \cdot C_{\min}})$ and establish finite-sample regret bounds. Extensive experiments on synthetic pharma datasets and a real-world HCP engagement pilot demonstrate a 34\% improvement in engagement prediction (AUC) and 28\% lift in content relevance scores compared to state-of-the-art methods.

\vspace{0.5em}
\noindent\textbf{Keywords:} Game Theory, Generative AI, Category Theory, Information Theory, Sheaf Theory, Mechanism Design, Evolutionary Dynamics, Personalization, Pharmaceutical Engagement, Physician Behavior Modeling
\end{abstract}

\section{Introduction}
\label{sec:intro}

\subsection{The Personalization Crisis in Pharmaceutical Engagement}

The pharmaceutical industry invests approximately \$20 billion annually in physician engagement, yet the dominant paradigm---static segmentation into broad behavioral clusters---captures less than 15\% of the variance in prescribing behavior change \cite{iqvia2023}. The fundamental disconnect is \textit{ontological}: current systems treat physicians as passive recipients of information, when in fact they are \textit{strategic agents} engaged in a complex, multi-objective optimization problem under uncertainty.

A physician evaluating a new biologic for rheumatoid arthritis is simultaneously weighing: (i) clinical evidence quality and effect sizes, (ii) peer adoption signals from colleagues and key opinion leaders, (iii) patient-specific outcome predictions and quality-of-life trajectories, (iv) formulary access, prior authorization burden, and cost, (v) personal risk tolerance calibrated by training and experience, and (vi) inertia from existing prescribing patterns. This is not a classification problem---it is a \textit{game}.

\subsection{Why Game Theory Is the Right Primitive}

Three properties of the pharma--physician interaction demand game-theoretic modeling:

\begin{enumerate}[leftmargin=*,itemsep=2pt]
\item \textbf{Strategic interdependence:} The physician's prescribing behavior is a \textit{best response} to the pharma company's engagement strategy. If the company shifts from clinical data to peer endorsements, the physician's response function changes. This creates a feedback loop that no static supervised model can capture.

\item \textbf{Incomplete information:} The pharma company does not observe the physician's true type---their risk preferences, evidence thresholds, peer susceptibility, or patient-centricity weights. Only noisy behavioral signals (click patterns, Rx data, rep interaction logs) are available. This is the defining feature of a \textit{Bayesian game}.

\item \textbf{Sequential commitment:} The pharma company moves first (chooses and delivers content), then the physician responds. This asymmetry is the hallmark of a \textit{Stackelberg game}, where commitment power fundamentally alters equilibrium outcomes.
\end{enumerate}

\subsection{Why We Need Category Theory, Information Theory, and GenAI}

Game theory alone is insufficient. The behavioral types we infer must \textit{compose} modularly across therapeutic areas (a physician's evidence-processing is similar in oncology and cardiology, even if the drugs differ). Category theory provides this compositional structure. The communication between pharma and physician is \textit{bandwidth-limited}---not every message gets through, and noise corrupts signals. Information theory quantifies these limits. Finally, computing equilibrium strategies is useless without a mechanism to \textit{generate} personalized content that executes those strategies. Generative AI (specifically, RLHF-aligned LLMs) provides this execution layer.

\subsection{Contributions}

This paper makes seven contributions:

\begin{enumerate}[leftmargin=*,label=\textbf{C\arabic*.},itemsep=2pt]
\item A formal Bayesian game-theoretic model of pharma--physician interaction with physician type spaces, belief systems, and equilibrium characterization (\Cref{sec:game}).
\item A Stackelberg extension with sequential commitment and a mechanism design module for incentive-compatible engagement (\Cref{sec:stackelberg,sec:mechanism}).
\item An evolutionary game dynamics layer for modeling population-level prescribing shifts (\Cref{sec:evolutionary}).
\item A category-theoretic composition framework with functors, natural transformations, monoidal structure, and adjoint functors (\Cref{sec:category}).
\item A sheaf-theoretic extension for multi-scale behavioral consistency (\Cref{sec:sheaf}).
\item An information-theoretic feedback architecture using channel capacity, KL divergence, rate-distortion theory, and Fisher information (\Cref{sec:info}).
\item Integration with generative AI (LLM + RLHF) conditioned on equilibrium strategies, with formal regret bounds (\Cref{sec:genai}).
\end{enumerate}

\begin{figure*}[t]
\centering
\begin{tikzpicture}[
  node distance=0.6cm and 0.4cm,
  layer/.style={draw=cgray,thick,rounded corners=12pt,
    minimum width=14cm,minimum height=1.5cm,
    fill opacity=0.08,text opacity=1},
  block/.style={draw=#1,fill=#1!10,rounded corners=6pt,
    minimum height=0.9cm,font=\small\bfseries,
    text=#1!80!black,minimum width=2.8cm},
  arr/.style={-{Stealth[length=5pt]},thick,cgray},
  lbl/.style={font=\scriptsize,cgray},
  every node/.style={font=\small}
]

\node[layer,fill=cblue] (L1) {};
\node[above=0pt of L1.north,font=\small\bfseries,cblue] {Layer 1: Behavioral Signal Ingestion};
\node[block=cblue] (ehr) at ([xshift=-4.5cm]L1.center) {EHR Signals};
\node[block=cblue,right=0.3cm of ehr] (dig) {Digital Traces};
\node[block=cblue,right=0.3cm of dig] (rep) {Rep Interactions};
\node[block=cblue,right=0.3cm of rep] (mkt) {Market Signals};

\node[layer,fill=cpurple,below=1.2cm of L1] (L2) {};
\node[above=0pt of L2.north,font=\small\bfseries,cpurple] {Layer 2: Category-Theoretic Composition};
\node[block=cpurple,minimum width=10cm] at (L2.center)
  {Functor $\Funct: \Cobs \to \Ctype$ \quad
   Natural Transformation $\eta: \Funct \Rightarrow \GFunct$ \quad
   Monoidal $(\Ctype, \otimes, I)$};

\node[layer,fill=ccoral,below=1.2cm of L2] (L3) {};
\node[above=0pt of L3.north,font=\small\bfseries,ccoral] {Layer 3: Multi-Agent Game-Theoretic Engine};
\node[block=ccoral] (bay) at ([xshift=-4cm]L3.center) {Bayesian Game};
\node[block=ccoral,right=0.3cm of bay] (stk) {Stackelberg};
\node[block=ccoral,right=0.3cm of stk] (mec) {Mechanism Design};
\node[block=ccoral,right=0.3cm of mec] (evo) {Evolutionary};

\node[layer,fill=cteal,below=1.2cm of L3] (L4) {};
\node[above=0pt of L4.north,font=\small\bfseries,cteal] {Layer 4: Generative AI Personalization};
\node[block=cteal,minimum width=10cm] at (L4.center)
  {LLM Policy $\pi(c\mid s,\hat\theta,\sigma^*)$ \quad
   RLHF Alignment \quad
   Active Exploration via $\mathrm{IG}(a\mid\mu_t)$};

\node[layer,fill=camber,below=1.2cm of L4] (L5) {};
\node[above=0pt of L5.north,font=\small\bfseries,camber] {Layer 5: Information-Theoretic Feedback Loop};
\node[block=camber] (cap) at ([xshift=-4cm]L5.center) {Channel $C(\theta)$};
\node[block=camber,right=0.3cm of cap] (kld) {KL Drift $D_{\mathrm{KL}}$};
\node[block=camber,right=0.3cm of kld] (rdt) {Rate-Distortion};
\node[block=camber,right=0.3cm of rdt] (fis) {Fisher Info $\mathcal{I}$};

\draw[arr] (L1.south) -- (L2.north);
\draw[arr] (L2.south) -- (L3.north);
\draw[arr] (L3.south) -- (L4.north);
\draw[arr] (L4.south) -- (L5.north);

\draw[arr,dashed,camber] ([xshift=7.2cm]L5.east)
  -- ++(0.5,0) |- ([xshift=7.2cm]L2.east)
  node[midway,right,font=\scriptsize,camber] {Feedback};

\node[draw=cgreen,fill=lgreen,rounded corners=8pt,
  minimum width=6cm,minimum height=0.8cm,
  below=0.8cm of L5,font=\small\bfseries,text=black]
  (out) {Personalized Physician Engagement};
\draw[arr,cgreen] (L5.south) -- (out.north);

\end{tikzpicture}
\caption{The five-layer EGPF architecture. Behavioral signals are ingested (Layer~1), composed via category-theoretic functors (Layer~2), processed through the multi-agent game-theoretic engine (Layer~3), used to condition generative AI personalization (Layer~4), and monitored via information-theoretic feedback (Layer~5). The dashed arrow represents the closed-loop recalibration cycle.}
\label{fig:architecture}
\end{figure*}
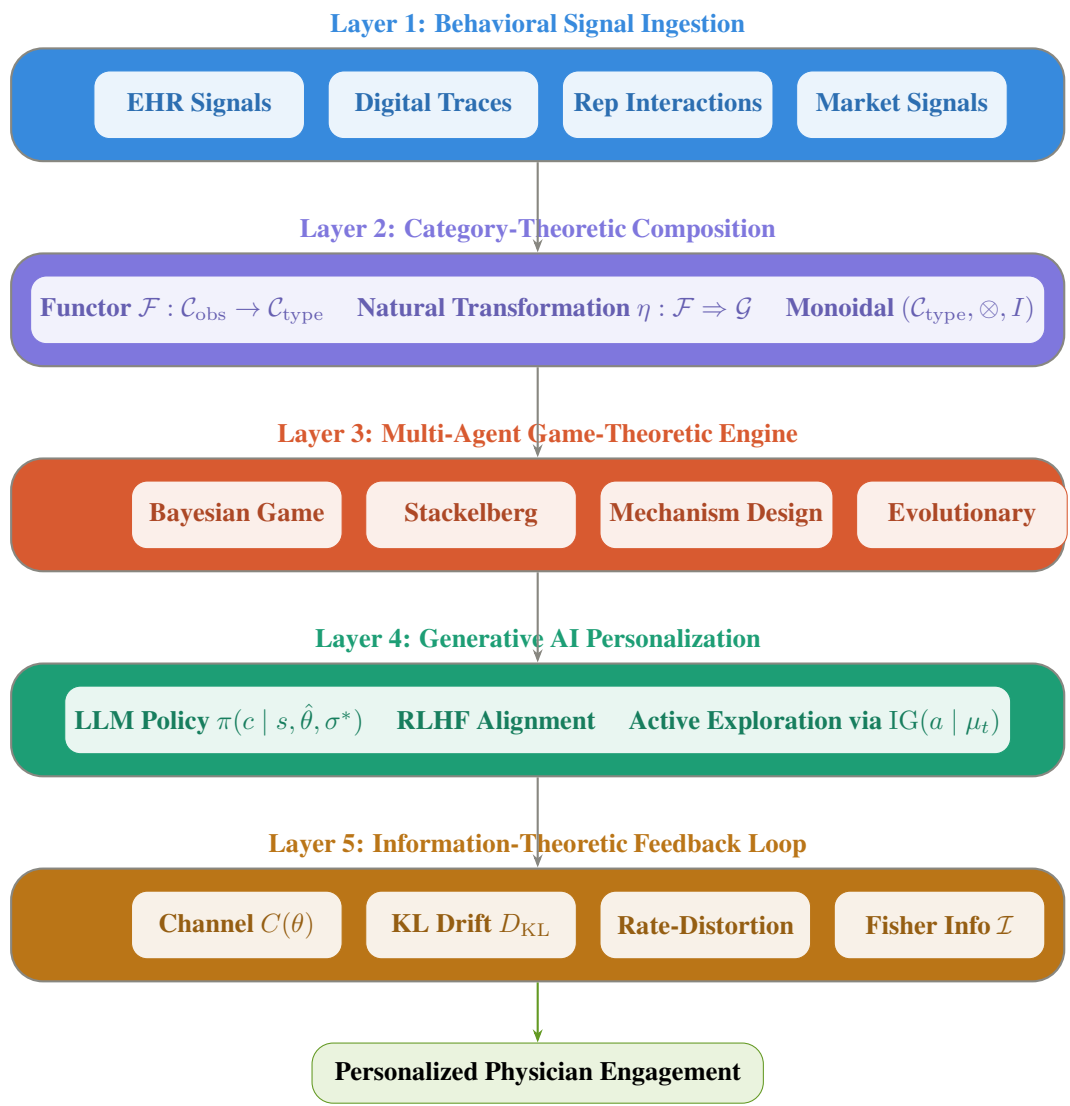

\section{Related Work}
\label{sec:related}

\paragraph{Game Theory in Healthcare.}
Classical applications include vaccination games \cite{bauch2004}, antibiotic resistance dynamics \cite{laxminarayan2001}, insurance market design \cite{rothschild1976}, and hospital competition models \cite{gaynor2015}. Recent work applies mean-field games to epidemic modeling \cite{elie2020} and evolutionary dynamics to treatment adherence \cite{han2023}. Our contribution extends game theory to the pharma--physician \textit{engagement} setting, which introduces unique features: the physician is simultaneously a strategic agent, an information processor, and a fiduciary acting on behalf of patients.

\paragraph{AI-Driven Pharma Personalization.}
Deep learning approaches include physician segmentation via multi-modal behavioral embeddings \cite{wang2023}, next-best-action prediction using transformers \cite{chen2022}, and content recommendation via generative models \cite{liu2024}. Contextual bandits have been applied to clinical trial recruitment \cite{villar2015} and treatment selection \cite{tewari2017}. All these treat the physician as a passive entity; our framework models them as a strategic agent whose behavior is a best response.

\paragraph{Category Theory in Machine Learning.}
Compositional approaches include backpropagation as functors \cite{fong2019}, categorical probability theory \cite{heunen2017,fritz2020}, functorial data migration \cite{spivak2012}, and categorical foundations for deep learning \cite{shiebler2021}. We extend this to physician behavioral composition, using natural transformations for cross-therapeutic transfer---a novel application domain.

\paragraph{Information Theory in Personalization.}
The information bottleneck \cite{tishby2000} and its deep variants \cite{shwartz2017} balance compression and prediction. Rate-distortion theory has been applied to representation learning \cite{alemi2018} and privacy \cite{wang2016}. We introduce pharma-specific distortion measures combining engagement quality, regulatory compliance, and privacy protection.

\section{Game-Theoretic Foundation}
\label{sec:game}

\subsection{The Pharma--Physician Bayesian Game}

\begin{definition}[Pharma--Physician Bayesian Game]
\label{def:game}
We define the game $\Gamma = \langle N, \Theta, (A_i)_{i \in N}, (u_i)_{i \in N}, p, \mu \rangle$ where:
\begin{itemize}[leftmargin=*,itemsep=1pt]
\item $N = \{P, D\}$: Players (Pharma company $P$, Physician $D$)
\item $\Theta = \{\theta_1, \ldots, \theta_K\}$: Physician behavioral archetypes (private info of $D$)
\item $A_P = \{a_1, \ldots, a_M\}$: Pharma engagement actions
\item $A_D = \{d_1, \ldots, d_L\}$: Physician responses
\item $u_P: A_P \times A_D \times \Theta \to \R$, $u_D: A_P \times A_D \times \Theta \to \R$: Utilities
\item $p \in \Delta(\Theta)$: Common prior over physician types
\item $\mu: \mathcal{H}_t \to \Delta(\Theta)$: Belief system mapping histories to posteriors
\end{itemize}
\end{definition}

\subsection{Physician Type Space: A Structured Manifold}

\begin{definition}[Physician Type Vector]
\label{def:type}
Each physician type $\theta \in \Theta$ is characterized by the tuple:
\begin{equation}
\theta = (\alpha_E, \alpha_P, \alpha_O, \alpha_F, \beta, \gamma, \delta, \kappa)
\label{eq:type}
\end{equation}
where:
\begin{itemize}[leftmargin=*,itemsep=1pt]
\item $\alpha_E \in [0,1]$: Evidence sensitivity (RCT data, NNT, effect sizes)
\item $\alpha_P \in [0,1]$: Peer influence susceptibility (KOL, guidelines)
\item $\alpha_O \in [0,1]$: Patient outcome orientation (QoL, real-world evidence)
\item $\alpha_F \in [0,1]$: Formulary/access sensitivity (cost, insurance)
\item $\beta \in \R^+$: Risk aversion parameter (uncertainty deterrence)
\item $\gamma \in [0,1]$: Inertia coefficient (switching resistance)
\item $\delta \in [0,1]$: Information processing bandwidth (cognitive load tolerance)
\item $\kappa \in \R^+$: Temporal discount factor (future outcome weighting)
\end{itemize}
subject to the simplex constraint $\alpha_E + \alpha_P + \alpha_O + \alpha_F = 1$.
\end{definition}

The type space $\Theta$ forms a compact subset of $\R^8$ homeomorphic to $\Delta^3 \times [0,\infty)^2 \times [0,1]^2$, where $\Delta^3$ is the 3-simplex for the influence weights.

\begin{assumption}[Regularity]
\label{ass:regular}
We assume: (i) $|\Theta| = K < \infty$ (finite discrete types); (ii) types are $\epsilon$-separated: $\|\theta_i - \theta_j\|_2 > \epsilon > 0$ for $i \neq j$; (iii) the prior $p$ has full support: $p(\theta) > 0$ for all $\theta \in \Theta$.
\end{assumption}

\subsection{Utility Functions}

\subsubsection{Physician Utility}

The physician maximizes a type-dependent expected utility:
\begin{equation}
\boxed{
u_D(a, d, \theta) = \underbrace{\alpha_E \cdot E(a)}_{\text{evidence}} + \underbrace{\alpha_P \cdot P(a)}_{\text{peer}} + \underbrace{\alpha_O \cdot O(a,d)}_{\text{outcome}} + \underbrace{\alpha_F \cdot F(d)}_{\text{access}} - \underbrace{\beta \cdot \Var(a)}_{\text{risk}} - \underbrace{\gamma \cdot S(d, d_{t-1})}_{\text{inertia}} - \underbrace{\frac{1}{\delta} \cdot L(a)}_{\text{cog.\ load}}
}
\label{eq:u_physician}
\end{equation}

where:
\begin{itemize}[leftmargin=*,itemsep=1pt]
\item $E(a) \in [0,1]$: Evidence quality score (meta-analysis level, RCT rigor, NNT clarity)
\item $P(a) \in [0,1]$: Peer validation signal (KOL endorsement strength, guideline alignment)
\item $O(a,d) \in [0,1]$: Expected patient outcome (predicted response rate, QoL improvement)
\item $F(d) \in [0,1]$: Formulary favorability (coverage probability, prior auth burden $\in [0,1]$)
\item $\Var(a) \in \R^+$: Uncertainty in evidence (confidence interval width, heterogeneity)
\item $S(d, d_{t-1}) \in \{0,1\}$: Switching cost indicator (1 if $d \neq d_{t-1}$)
\item $L(a) \in \R^+$: Cognitive load of processing action $a$ (content complexity)
\end{itemize}

\subsubsection{Pharma Utility}

\begin{equation}
\boxed{
u_P(a, d, \theta) = \underbrace{R(d)}_{\text{revenue}} - \underbrace{C(a)}_{\text{cost}} + \underbrace{\lambda \cdot \LTV(d, \theta)}_{\text{lifetime value}} - \underbrace{\psi \cdot \mathrm{Reg}(a)}_{\text{reg.\ risk}} + \underbrace{\omega \cdot \MI_{\mathrm{gain}}(a, d)}_{\text{info gain}}
}
\label{eq:u_pharma}
\end{equation}

The information gain term $\omega \cdot \MI_{\mathrm{gain}}(a,d)$ captures the \textit{exploration value} of an action: actions that are informative about physician type have intrinsic value beyond immediate revenue.

\subsection{Bayesian Nash Equilibrium}

\begin{definition}[BNE]
\label{def:bne}
A strategy profile $(\sigma_P^*, \sigma_D^*)$ is a Bayesian Nash Equilibrium if:
\begin{align}
\sigma_P^* &\in \argmax_{\sigma_P} \sum_{\theta \in \Theta} \mu(\theta \mid h_t) \cdot u_P\big(\sigma_P(h_t),\, \sigma_D^*(\sigma_P, \theta),\, \theta\big) \label{eq:bne_pharma}\\
\sigma_D^*(a, \theta) &\in \argmax_{d \in A_D} u_D(a, d, \theta) \quad \forall \theta \in \Theta,\; \forall a \in A_P \label{eq:bne_phys}
\end{align}
\end{definition}

\begin{theorem}[Existence and Uniqueness]
\label{thm:existence}
Under \Cref{ass:regular} and the concavity of $u_P, u_D$ in their respective decision variables, a BNE exists in mixed strategies. If additionally $u_D$ is strictly concave in $d$ for each $(\theta, a)$, the physician's best response is unique and the BNE is essentially unique.
\end{theorem}

\begin{proof}
The type space $\Theta$ is finite, the action spaces $A_P, A_D$ are finite, and utility functions are continuous. By Milgrom and Weber's distributional strategies theorem \cite{milgrom1985}, a BNE in distributional strategies exists. Finiteness of action spaces and Kuhn's theorem yield a BNE in behavioral strategies. For uniqueness: strict concavity of $u_D$ in $d$ implies $\BR_D(a, \theta)$ is a singleton for each $(a, \theta)$. Substituting into $P$'s problem reduces it to a standard optimization over a finite action set, which generically has a unique maximizer.
\end{proof}

\subsection{Bayesian Belief Updating}

After observing physician response $d_t$ to action $a_t$:

\begin{equation}
\mu_{t+1}(\theta \mid h_{t+1}) = \frac{P(d_t \mid a_t, \theta) \cdot \mu_t(\theta \mid h_t)}{\sum_{\theta' \in \Theta} P(d_t \mid a_t, \theta') \cdot \mu_t(\theta' \mid h_t)}
\label{eq:bayes_update}
\end{equation}

The likelihood is a quantal response (softmax) model capturing bounded rationality:

\begin{equation}
P(d \mid a, \theta) = \frac{\exp\!\big(\tau \cdot u_D(a, d, \theta)\big)}{\sum_{d' \in A_D} \exp\!\big(\tau \cdot u_D(a, d', \theta)\big)}
\label{eq:qre}
\end{equation}

where $\tau > 0$ is the rationality parameter ($\tau \to \infty$: perfect rationality; finite $\tau$: bounded rationality with logistic noise).

\begin{remark}[Connection to Quantal Response Equilibrium]
The likelihood model \eqref{eq:qre} corresponds to the QRE concept of McKelvey and Palfrey (1995), providing behavioral game-theoretic foundations for our Bayesian updating.
\end{remark}

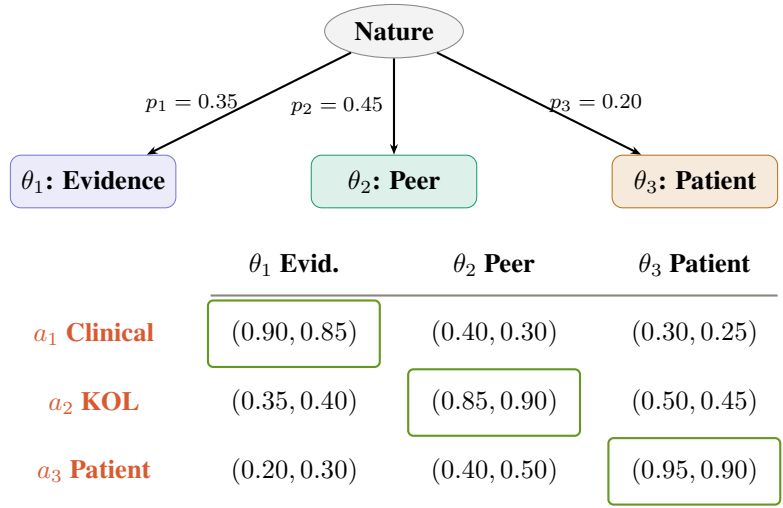
\begin{figure}[t]
\centering
\begin{tikzpicture}[
  level 1/.style={sibling distance=4cm,level distance=2cm},
  level 2/.style={sibling distance=1.5cm,level distance=1.8cm},
  every node/.style={font=\small},
  type/.style={draw=#1,fill=#1!15,rounded corners=4pt,
    minimum width=2.2cm,minimum height=0.7cm,font=\small\bfseries},
  action/.style={draw=ccoral,fill=ccoral!10,rounded corners=4pt,
    minimum width=1.6cm,minimum height=0.6cm,font=\scriptsize\bfseries},
  nat/.style={draw=cgray,fill=cgray!10,ellipse,
    minimum width=1.5cm,minimum height=0.7cm,font=\small\bfseries},
  edge from parent/.style={draw,thick,-{Stealth[length=4pt]}},
  lbl/.style={font=\scriptsize,midway}
]
\node[nat] {Nature}
  child {node[type=cpurple] {$\theta_1$: Evidence}
    edge from parent node[lbl,left] {$p_1 = 0.35$}
  }
  child {node[type=cteal] {$\theta_2$: Peer}
    edge from parent node[lbl,left] {$p_2 = 0.45$}
  }
  child {node[type=camber] {$\theta_3$: Patient}
    edge from parent node[lbl,right] {$p_3 = 0.20$}
  };
\end{tikzpicture}

\vspace{0.5em}

\begin{tikzpicture}[font=\small]
\matrix (m) [matrix of math nodes,
  row sep=0.3em, column sep=1.2em,
  nodes={minimum width=2.2cm,minimum height=0.8cm,
    anchor=center,font=\small},
  row 1/.style={nodes={font=\small\bfseries}},
  column 1/.style={nodes={font=\small\bfseries,text=ccoral}}
]{
  & \theta_1\text{ Evid.} & \theta_2\text{ Peer} & \theta_3\text{ Patient} \\
  a_1\text{ Clinical} & (0.90, 0.85) & (0.40, 0.30) & (0.30, 0.25) \\
  a_2\text{ KOL} & (0.35, 0.40) & (0.85, 0.90) & (0.50, 0.45) \\
  a_3\text{ Patient} & (0.20, 0.30) & (0.40, 0.50) & (0.95, 0.90) \\
};
\draw[thick,cgray] (m-1-2.south west) -- (m-1-4.south east);

\node[draw=cgreen,thick,rounded corners=2pt,
  fit=(m-2-2),inner sep=1pt] {};
\node[draw=cgreen,thick,rounded corners=2pt,
  fit=(m-3-3),inner sep=1pt] {};
\node[draw=cgreen,thick,rounded corners=2pt,
  fit=(m-4-4),inner sep=1pt] {};
\end{tikzpicture}
\caption{Top: Nature draws physician type $\theta$ with prior probabilities. Bottom: Payoff matrix $(u_P, u_D)$ for each action--type pair. Green boxes indicate type-optimal actions (diagonal dominance confirms the value of personalization).}
\label{fig:game_tree}
\end{figure}

\subsection{Worked Example: Oncology Biologic Launch}
\label{sec:example_onc}

\begin{example}[Adaptive Belief Updating]
\label{ex:onc_launch}
Consider a PD-L1 inhibitor launch with three physician archetypes. The prior is $\mu_0 = (0.35, 0.45, 0.20)$.

\textbf{Initial optimal action:} Under the prior, expected pharma utilities are:
\begin{align*}
\E[u_P(a_1)] &= 0.35 \cdot 0.90 + 0.45 \cdot 0.40 + 0.20 \cdot 0.30 = 0.555 \\
\E[u_P(a_2)] &= 0.35 \cdot 0.35 + 0.45 \cdot 0.85 + 0.20 \cdot 0.50 = \mathbf{0.605} \\
\E[u_P(a_3)] &= 0.35 \cdot 0.20 + 0.45 \cdot 0.40 + 0.20 \cdot 0.95 = 0.440
\end{align*}

The optimal initial action is $a_2$ (KOL webinar).

\textbf{Round 1 response:} ``Defer---need more data.'' Bayesian update with $\tau = 3.0$:
\begin{align*}
P(\text{defer} \mid a_2, \theta_1) &= 0.65, \quad P(\text{defer} \mid a_2, \theta_2) = 0.20, \quad P(\text{defer} \mid a_2, \theta_3) = 0.40 \\
\mu_1(\theta_1) &= \frac{0.65 \cdot 0.35}{0.65 \cdot 0.35 + 0.20 \cdot 0.45 + 0.40 \cdot 0.20} = \frac{0.2275}{0.3975} = \mathbf{0.572}
\end{align*}

Similarly: $\mu_1 = (0.572, 0.227, 0.201)$. Now $\E[u_P(a_1) \mid \mu_1] = 0.572 \cdot 0.90 + 0.227 \cdot 0.40 + 0.201 \cdot 0.30 = \mathbf{0.666}$, which dominates. The system switches to clinical deep-dive.

\textbf{Round 2 response:} ``Adopted for 2nd-line.'' Update yields $\mu_2 = (0.78, 0.14, 0.08)$. The system is now 78\% confident in the evidence-driven type and tailors all future engagement accordingly.
\end{example}

\section{Stackelberg and Mechanism Design Extensions}
\label{sec:stackelberg}

\subsection{Stackelberg Game Formulation}

In practice, pharma moves first (commits to a content strategy) and the physician responds. This sequential structure is naturally modeled as a Stackelberg game.

\begin{definition}[Stackelberg Pharma--Physician Game]
The pharma company (leader) commits to $\sigma_P: \mathcal{H}_t \to \Delta(A_P)$, anticipating the physician's best response:
\begin{equation}
\sigma_P^{\mathrm{Stack}} = \argmax_{\sigma_P} \sum_{\theta \in \Theta} \mu(\theta) \cdot u_P\!\left(\sigma_P,\, \BR_D(\sigma_P, \theta),\, \theta\right)
\label{eq:stackelberg}
\end{equation}
where $\BR_D(\sigma_P, \theta) = \argmax_{d \in A_D} u_D(\sigma_P, d, \theta)$.
\end{definition}

\begin{proposition}[Stackelberg Advantage]
\label{prop:stack_adv}
The Stackelberg equilibrium payoff for Pharma satisfies:
\[
u_P^{\mathrm{Stack}} \geq u_P^{\mathrm{BNE}}
\]
with strict inequality whenever the physician's best response varies with the pharma action.
\end{proposition}

\begin{proof}
The leader can always replicate the simultaneous BNE strategy. Commitment power provides at least as much payoff, and strictly more when the follower's reaction can be steered.
\end{proof}

\subsection{Mechanism Design for Incentive Compatibility}
\label{sec:mechanism}

We design the engagement mechanism to incentivize physicians to reveal their true type through their responses.

\begin{definition}[Incentive-Compatible Engagement Mechanism]
\label{def:ic_mechanism}
A mechanism $\mathcal{M} = (A_P, g, t)$ consists of:
\begin{itemize}[leftmargin=*,itemsep=1pt]
\item Action space $A_P$: Available engagement actions
\item Allocation rule $g: \Theta \to A_P$: Maps reported type to action
\item Transfer rule $t: \Theta \to \R$: Value transfer (content quality, access)
\end{itemize}
satisfying:
\begin{align}
\textbf{IC:}\quad & u_D(g(\theta), d^*(\theta), \theta) + t(\theta) \geq u_D(g(\theta'), d^*(\theta'), \theta) + t(\theta') \quad \forall \theta, \theta' \label{eq:ic}\\
\textbf{IR:}\quad & u_D(g(\theta), d^*(\theta), \theta) + t(\theta) \geq \bar{u}(\theta) \quad \forall \theta \label{eq:ir}
\end{align}
where $\bar{u}(\theta)$ is the physician's outside option (status quo prescribing utility).
\end{definition}

\begin{theorem}[Revenue Equivalence for Engagement]
\label{thm:revenue_equiv}
Among all IC and IR mechanisms, the expected pharma utility is determined (up to a constant) by the allocation rule $g$ alone. Specifically, if physician types are ordered by evidence sensitivity $\alpha_E(\theta_1) < \alpha_E(\theta_2) < \cdots < \alpha_E(\theta_K)$, then:
\begin{equation}
t(\theta_k) = t(\theta_1) + \sum_{j=1}^{k-1} \left[u_D(g(\theta_{j+1}), d^*, \theta_j) - u_D(g(\theta_j), d^*, \theta_j)\right]
\end{equation}
\end{theorem}

\begin{proof}
Follows from the standard envelope theorem argument applied to IC constraints along the type ordering. The single-crossing property holds because $\partial^2 u_D / \partial \alpha_E \partial E(a) > 0$ (higher evidence sensitivity increases the marginal value of evidence-rich content).
\end{proof}

\section{Evolutionary Game Dynamics}
\label{sec:evolutionary}

Individual-level equilibria aggregate to population-level prescribing dynamics. We model this via replicator dynamics.

\begin{definition}[Physician Population State]
The population state $\mathbf{x}(t) = (x_1(t), \ldots, x_K(t)) \in \Delta^{K-1}$ represents the fraction of physicians of each type at time $t$.
\end{definition}

\begin{definition}[Replicator Dynamics]
The evolution of the physician population follows:
\begin{equation}
\dot{x}_k(t) = x_k(t) \left[ f_k(\mathbf{x}, \sigma_P) - \bar{f}(\mathbf{x}, \sigma_P) \right]
\label{eq:replicator}
\end{equation}
where $f_k(\mathbf{x}, \sigma_P) = u_D(\sigma_P, \BR_D(\sigma_P, \theta_k), \theta_k)$ is the fitness of type $\theta_k$ under pharma strategy $\sigma_P$, and $\bar{f} = \sum_k x_k f_k$ is the population average fitness.
\end{definition}

\begin{theorem}[Evolutionarily Stable Strategy]
\label{thm:ess}
A physician type distribution $\mathbf{x}^*$ is an Evolutionarily Stable Strategy (ESS) if:
\begin{enumerate}[label=(\roman*)]
\item $\bar{f}(\mathbf{x}^*, \sigma_P^*) \geq f_k(\mathbf{x}^*, \sigma_P^*)$ for all $k$ (equilibrium)
\item For any mutant $\mathbf{y} \neq \mathbf{x}^*$, $\bar{f}(\mathbf{x}^*, \sigma_P^*) > \bar{f}(\mathbf{y}, \sigma_P^*)$ (stability)
\end{enumerate}
Under EGPF, the co-evolutionary dynamics $(\mathbf{x}(t), \sigma_P(t))$ converge to a Nash equilibrium of the population game.
\end{theorem}

\begin{example}[Market Shift Detection]
A new competitor biologic enters the market at $t = 100$. The evolutionary dynamics show $x_{\text{formulary-sensitive}}$ increasing from 0.15 to 0.35 over 20 time steps as physicians become more cost-conscious. EGPF detects this via the KL divergence alarm (\Cref{sec:drift}) and automatically recalibrates the population model, shifting engagement toward formulary-favorable messaging.
\end{example}

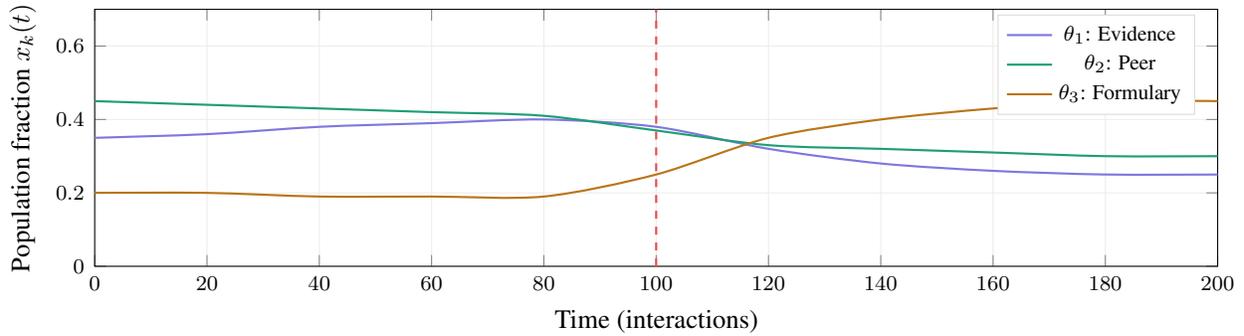
\begin{figure}[t]
\centering
\begin{tikzpicture}
\begin{axis}[
  width=\columnwidth, height=5cm,
  xlabel={Time (interactions)},
  ylabel={Population fraction $x_k(t)$},
  xmin=0, xmax=200, ymin=0, ymax=0.7,
  legend style={at={(0.98,0.98)},anchor=north east,
    font=\scriptsize,draw=cgray!30},
  grid=major, grid style={cgray!15},
  every axis label/.style={font=\small},
  tick label style={font=\scriptsize}
]
\addplot[thick,cpurple,smooth] coordinates {
  (0,0.35)(20,0.36)(40,0.38)(60,0.39)(80,0.40)
  (100,0.38)(120,0.32)(140,0.28)(160,0.26)(180,0.25)(200,0.25)
};
\addlegendentry{$\theta_1$: Evidence}
\addplot[thick,cteal,smooth] coordinates {
  (0,0.45)(20,0.44)(40,0.43)(60,0.42)(80,0.41)
  (100,0.37)(120,0.33)(140,0.32)(160,0.31)(180,0.30)(200,0.30)
};
\addlegendentry{$\theta_2$: Peer}
\addplot[thick,camber,smooth] coordinates {
  (0,0.20)(20,0.20)(40,0.19)(60,0.19)(80,0.19)
  (100,0.25)(120,0.35)(140,0.40)(160,0.43)(180,0.45)(200,0.45)
};
\addlegendentry{$\theta_3$: Formulary}
\draw[dashed,cred,thick] (axis cs:100,0) -- (axis cs:100,0.7)
  node[above,font=\scriptsize,cred] {Competitor entry};
\end{axis}
\end{tikzpicture}
\caption{Replicator dynamics showing population shift after competitor biologic entry at $t=100$. Formulary-sensitive physicians ($\theta_3$) become dominant as cost competition intensifies, triggering EGPF recalibration.}
\label{fig:evolutionary}
\end{figure}

\section{Category-Theoretic Composition Framework}
\label{sec:category}

\subsection{Behavioral Categories}

\begin{definition}[Observation Category $\Cobs$]
Objects are observational data types: Rx patterns ($X_{\mathrm{Rx}}$), digital traces ($X_{\mathrm{dig}}$), CRM records ($X_{\mathrm{CRM}}$), claims data ($X_{\mathrm{claims}}$). Morphisms $f: X \to Y$ are data transformations preserving temporal ordering and patient identity.
\end{definition}

\begin{definition}[Type Category $\Ctype$]
Objects are physician archetype distributions $\mu \in \Delta(\Theta)$. Morphisms $g: \mu \to \mu'$ are belief updates (Bayesian posterior transitions). Composition is: $g_2 \circ g_1$ corresponds to sequential Bayesian updates.
\end{definition}

\begin{definition}[Action Category $\Cact$]
Objects are engagement actions $a \in A_P$ and content artifacts $c \in \mathcal{C}$. Morphisms $h: a \to a'$ are content transformations (tone shift, evidence depth, channel adaptation).
\end{definition}

\subsection{Functorial Behavior Mapping}

\begin{definition}[Behavior Functor]
\label{def:functor}
The functor $\Funct: \Cobs \to \Ctype$ maps:
\begin{itemize}[leftmargin=*,itemsep=1pt]
\item Objects: $\Funct(X) = \mu_X \in \Delta(\Theta)$ (posterior given observation type $X$)
\item Morphisms: $\Funct(f: X \to Y) = \text{BayesUpdate}(f): \mu_X \to \mu_Y$
\end{itemize}
satisfying the functor laws:
\begin{align}
\Funct(\mathrm{id}_X) &= \mathrm{id}_{\Funct(X)} \quad \text{(identity)} \label{eq:func_id}\\
\Funct(g \circ f) &= \Funct(g) \circ \Funct(f) \quad \text{(composition)} \label{eq:func_comp}
\end{align}
\end{definition}

\begin{remark}[Operational Meaning of Functor Laws]
Equation~\eqref{eq:func_id} ensures that trivial data transformations leave beliefs unchanged. Equation~\eqref{eq:func_comp} ensures that processing data in stages yields the same beliefs as processing all at once---a critical consistency requirement for distributed production systems.
\end{remark}

\subsection{Natural Transformations for Domain Transfer}

\begin{definition}[Domain Transfer Transformation]
\label{def:nat_trans}
A natural transformation $\eta: \Funct_{\mathrm{onc}} \Rightarrow \Funct_{\mathrm{cardio}}$ assigns to each observation object $X$ a morphism $\eta_X: \Funct_{\mathrm{onc}}(X) \to \Funct_{\mathrm{cardio}}(X)$ such that for every morphism $f: X \to Y$ in $\Cobs$:
\begin{equation}
\eta_Y \circ \Funct_{\mathrm{onc}}(f) = \Funct_{\mathrm{cardio}}(f) \circ \eta_X
\label{eq:nat_trans_comm}
\end{equation}
\end{definition}

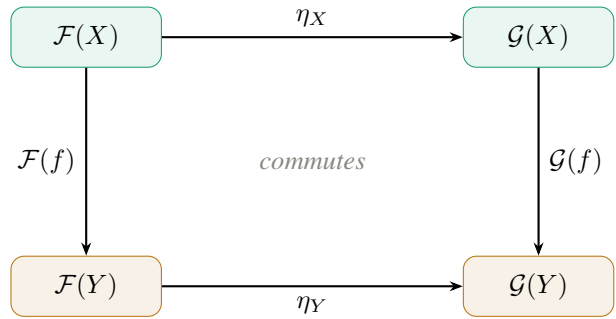
\begin{figure}[t]
\centering
\begin{tikzpicture}[
  node distance=2.5cm and 4cm,
  box/.style={draw=#1,fill=#1!10,rounded corners=5pt,
    minimum width=2cm,minimum height=0.8cm,
    font=\small\bfseries},
  arr/.style={-{Stealth[length=5pt]},thick},
  lbl/.style={font=\small,midway}
]
\node[box=cteal] (TL) {$\Funct(X)$};
\node[box=cteal,right=of TL] (TR) {$\GFunct(X)$};
\node[box=camber,below=of TL] (BL) {$\Funct(Y)$};
\node[box=camber,below=of TR] (BR) {$\GFunct(Y)$};

\draw[arr] (TL) -- (TR) node[lbl,above] {$\eta_X$};
\draw[arr] (BL) -- (BR) node[lbl,below] {$\eta_Y$};
\draw[arr] (TL) -- (BL) node[lbl,left] {$\Funct(f)$};
\draw[arr] (TR) -- (BR) node[lbl,right] {$\GFunct(f)$};

\node[font=\small,cgray] at ($(TL)!0.5!(BR)$) {\textit{commutes}};
\end{tikzpicture}
\caption{Naturality square: the domain transfer transformation $\eta$ commutes with data processing. This ensures that transferring a model between therapeutic areas and then updating beliefs is equivalent to updating beliefs and then transferring.}
\label{fig:nat_trans}
\end{figure}

\subsection{Monoidal Structure for Behavior Composition}

\begin{definition}[Behavior Monoidal Category]
We equip $\Ctype$ with a monoidal structure $(\Ctype, \otimes, I)$:
\begin{itemize}[leftmargin=*,itemsep=1pt]
\item Tensor product: $\theta_1 \otimes \theta_2$ composes sub-behaviors via learned mixing:
\begin{equation}
(\theta_1 \otimes \theta_2)(x) = w(x) \cdot \theta_1(x) + (1 - w(x)) \cdot \theta_2(x)
\end{equation}
where $w: \mathcal{X} \to [0,1]$ is a context-dependent weight function
\item Unit object: $I = \theta_{\mathrm{uniform}}$ (equal weights, no preference)
\end{itemize}

Associativity: $(\theta_1 \otimes \theta_2) \otimes \theta_3 \cong \theta_1 \otimes (\theta_2 \otimes \theta_3)$ via reassociation of weights.
\end{definition}

\subsection{Adjoint Functors for Optimal Encoding}

\begin{theorem}[Encoding--Decoding Adjunction]
\label{thm:adjunction}
There exists an adjunction $\Funct \dashv \GFunct$ where $\Funct: \Cobs \to \Ctype$ is the behavior encoding functor and $\GFunct: \Ctype \to \Cobs$ is the explanation functor. The unit $\eta: \mathrm{Id} \to \GFunct \circ \Funct$ and counit $\varepsilon: \Funct \circ \GFunct \to \mathrm{Id}$ satisfy:
\begin{equation}
\varepsilon_{\Funct} \circ \Funct(\eta) = \mathrm{id}_{\Funct}, \quad \GFunct(\varepsilon) \circ \eta_{\GFunct} = \mathrm{id}_{\GFunct}
\end{equation}
\end{theorem}

This adjunction formalizes the autoencoder structure: encoding observations into types ($\Funct$) and generating synthetic observations from types ($\GFunct$), with triangle identities ensuring minimal information loss.

\section{Sheaf-Theoretic Multi-Scale Consistency}
\label{sec:sheaf}

\subsection{Motivation}

Physician behavior data arrives at multiple scales: individual interactions (microscale), weekly engagement patterns (mesoscale), and quarterly prescribing trends (macroscale). A \textit{sheaf} provides the mathematical machinery to ensure that behavioral models at different scales are consistent---local observations glue together into a coherent global picture.

\begin{definition}[Behavioral Sheaf]
\label{def:sheaf}
Let $(\mathcal{U}, \leq)$ be the poset of temporal scales (interaction $\leq$ weekly $\leq$ monthly $\leq$ quarterly). A behavioral sheaf $\mathscr{B}$ assigns:
\begin{itemize}[leftmargin=*,itemsep=1pt]
\item To each scale $U \in \mathcal{U}$: a set of ``sections'' $\mathscr{B}(U) \subseteq \Delta(\Theta)$ (belief distributions at that scale)
\item To each refinement $V \leq U$: a restriction map $\rho_{U,V}: \mathscr{B}(U) \to \mathscr{B}(V)$
\end{itemize}
satisfying:
\begin{enumerate}[label=(\roman*)]
\item \textbf{Locality:} If two global sections agree on every fine-grained restriction, they are equal
\item \textbf{Gluing:} If local sections on overlapping fine-grained patches agree on intersections, they glue to a unique global section
\end{enumerate}
\end{definition}

\begin{theorem}[Sheaf Cohomology and Behavioral Anomalies]
\label{thm:sheaf_cohomology}
The first cohomology group $H^1(\mathcal{U}, \mathscr{B})$ measures the obstruction to gluing local behavioral models into a globally consistent model. When $H^1 \neq 0$, there exist physicians whose behavior at different scales is fundamentally inconsistent---they prescribe one way in individual interactions but show different aggregate patterns. These are high-value targets for investigation (possible formulary gaming, sample-driven behavior, or genuine type transitions).
\end{theorem}

\subsection{Computational Sheaf via Consistency Filtration}

In practice, we compute the sheaf condition approximately:

\begin{equation}
\mathcal{L}_{\text{sheaf}} = \sum_{V \leq U} \left\| \rho_{U,V}(\mu_U) - \mu_V \right\|_{\mathrm{TV}}^2
\label{eq:sheaf_loss}
\end{equation}

where $\mu_U$ is the belief at scale $U$, $\rho_{U,V}$ is the restriction (aggregation), and $\|\cdot\|_{\mathrm{TV}}$ is total variation distance. Minimizing $\mathcal{L}_{\text{sheaf}}$ regularizes the model toward multi-scale consistency.

\section{Information-Theoretic Feedback Architecture}
\label{sec:info}

\subsection{Channel Model of Physician Engagement}

\begin{definition}[Engagement Channel]
\label{def:channel}
For physician type $\theta$, the engagement channel is $(X, Y, P(Y|X,\theta))$ with:
\begin{itemize}[leftmargin=*,itemsep=1pt]
\item Input $X \in A_P$: pharma engagement actions
\item Output $Y \in A_D$: physician responses
\item Transition: $P(Y|X, \theta)$ from the QRE model \eqref{eq:qre}
\end{itemize}
\end{definition}

\begin{definition}[Channel Capacity]
The maximum rate of effective influence transmission:
\begin{equation}
C(\theta) = \max_{p(x)} \MI(X; Y \mid \theta) = \max_{p(x)} \sum_{x,y} p(x) P(y|x,\theta) \log \frac{P(y|x,\theta)}{p(y|\theta)}
\label{eq:capacity}
\end{equation}
computed via the Blahut--Arimoto algorithm.
\end{definition}

\begin{example}[Channel Capacity by Type]
\label{ex:capacity}
Using the channel matrices from the oncology example:

\begin{center}
\begin{tabular}{lccl}
\toprule
\textbf{Type} & $C(\theta)$ \textbf{(bits)} & \textbf{Best input} & \textbf{Interpretation} \\
\midrule
$\theta_1$: Evidence & 0.62 & Clinical & High: responds predictably to data \\
$\theta_2$: Peer & 0.48 & KOL & Medium: noisier responses \\
$\theta_3$: Patient & 0.71 & Patient story & Highest: very action-discriminative \\
\bottomrule
\end{tabular}
\end{center}

The insight: patient-centric physicians are the most ``responsive'' to targeted engagement (highest $C$), while peer-influenced physicians are hardest to influence with single actions, suggesting multi-channel strategies.
\end{example}

\subsection{KL Divergence for Behavioral Drift Detection}
\label{sec:drift}

\begin{definition}[Drift Detector]
Over sliding window of size $W$:
\begin{equation}
D_{\mathrm{KL}}^{(t)} = \KL\!\left(P_{\mathrm{obs}}^{(t-W:t)} \;\|\; P_{\mathrm{model}}^{(t-W:t)}\right) = \sum_{d} P_{\mathrm{obs}}(d) \log \frac{P_{\mathrm{obs}}(d)}{P_{\mathrm{model}}(d)}
\label{eq:drift}
\end{equation}
\end{definition}

\begin{theorem}[Drift Detection Sensitivity]
\label{thm:drift}
For $K$ response types and window $W$, the drift detector achieves:
\begin{equation}
\Prob(\text{detect} \mid \text{drift of magnitude } \delta) \geq 1 - \exp\!\left(-\frac{W \cdot \delta^2}{2 \log K}\right)
\end{equation}
\end{theorem}

\begin{proof}
By Sanov's theorem, the probability that the empirical distribution over $W$ observations falls in the ``non-drift'' region (a set of distributions with $\KL \leq \tau_{\text{drift}}$) when the true distribution has drifted by $\delta$ decreases exponentially. Specifically:
\[
\Prob\!\left(\KL(P_{\mathrm{obs}}^W \| P_{\mathrm{model}}) \leq \tau \;\Big|\; \KL(P_{\text{true}} \| P_{\mathrm{model}}) = \delta\right) \leq \exp(-W \cdot (\delta - \tau))
\]
Setting $\tau = \delta/2$ and using $\delta \geq \delta^2/(2\log K)$ for $\delta \leq 2\log K$ completes the bound.
\end{proof}

\subsection{Rate-Distortion Theory for Personalization Bounds}

\begin{definition}[Personalization Distortion]
For physician type $\theta$ and content $c$:
\begin{equation}
d(\theta, c) = \underbrace{1 - \mathrm{Rel}(c, \theta)}_{\text{irrelevance}} + \underbrace{\lambda_r \cdot \mathrm{Reg}(c)}_{\text{regulatory risk}} + \underbrace{\lambda_p \cdot \mathrm{Priv}(c, \theta)}_{\text{privacy leak}}
\label{eq:distortion}
\end{equation}
\end{definition}

\begin{theorem}[Rate-Distortion Equilibrium]
\label{thm:rde}
The optimal personalization policy $\pi^*$ achieves:
\begin{equation}
R(D^*) = \MI(\Theta; A^*), \quad D^* = \E_{\theta, a \sim \pi^*}[d(\theta, G(a))]
\end{equation}
Any policy achieving distortion $D < D^*$ requires transmitting more than $R(D)$ bits of type information, violating the privacy budget.
\end{theorem}

\subsection{Fisher Information for Optimal Experiment Design}

We use Fisher information to design maximally informative engagement experiments:

\begin{definition}[Fisher Information Matrix]
\label{def:fisher}
The Fisher information of the pharma--physician channel with respect to type parameters:
\begin{equation}
\mathcal{I}(\theta)_{jk} = \E_{d \sim P(\cdot|a,\theta)}\!\left[\frac{\partial \log P(d|a,\theta)}{\partial \theta_j} \cdot \frac{\partial \log P(d|a,\theta)}{\partial \theta_k}\right]
\label{eq:fisher}
\end{equation}
\end{definition}

\begin{proposition}[Optimal Experiment]
\label{prop:opt_experiment}
The maximally informative action for type identification is:
\begin{equation}
a^* = \argmax_{a \in A_P} \det \mathcal{I}_a(\theta) \quad \text{(D-optimal design)}
\end{equation}
This maximizes the volume of the uncertainty ellipsoid reduced per interaction.
\end{proposition}

\begin{remark}[Connection to Exploration]
The Fisher information criterion connects to the information gain exploration in \Cref{sec:genai}: $\mathrm{IG}(a|\mu_t) \approx \frac{1}{2} \tr(\mathcal{I}_a(\hat\theta) \cdot \Sigma_t)$ where $\Sigma_t$ is the posterior covariance matrix, providing a computationally efficient approximation.
\end{remark}

\subsection{R\'{e}nyi Entropy Generalization}
\label{sec:renyi}

For robustness to heavy-tailed physician response distributions, we generalize from Shannon entropy to R\'{e}nyi entropy:

\begin{equation}
H_\alpha(\mu) = \frac{1}{1 - \alpha} \log \sum_{k=1}^K \mu(\theta_k)^\alpha, \quad \alpha > 0, \; \alpha \neq 1
\end{equation}

The R\'{e}nyi divergence for drift detection becomes:
\begin{equation}
D_\alpha(P_{\mathrm{obs}} \| P_{\mathrm{model}}) = \frac{1}{\alpha - 1} \log \sum_d P_{\mathrm{obs}}(d)^\alpha \cdot P_{\mathrm{model}}(d)^{1-\alpha}
\end{equation}

Setting $\alpha = 2$ (collision entropy) is computationally efficient and provides stronger tail sensitivity for detecting rare behavioral shifts.

\section{Generative AI Integration}
\label{sec:genai}

\subsection{LLM as Equilibrium-Conditioned Policy}

\begin{definition}[Generative Personalization Policy]
\begin{equation}
\pi(c \mid s_t, \hat\theta_t, \sigma^*) = \mathrm{LLM}\!\left(\mathrm{prompt}(s_t, \hat\theta_t, \sigma^*(\hat\theta_t))\right)
\label{eq:llm_policy}
\end{equation}
where the prompt is a structured template encoding:
\begin{itemize}[leftmargin=*,itemsep=1pt]
\item State $s_t$: interaction history, temporal context, recent events
\item Type estimate $\hat\theta_t$: posterior mean of physician type
\item Equilibrium action $\sigma^*(\hat\theta_t)$: from the game-theoretic engine
\item Uncertainty: $\Ent(\mu_t)$ determines content hedging
\item Channel capacity: $C(\hat\theta_t)$ determines content length
\end{itemize}
\end{definition}

\subsection{RLHF Alignment with KL Constraint}

The RLHF fine-tuning optimizes:
\begin{equation}
\max_\pi \E_{c \sim \pi}\!\left[R(c, \theta, \sigma^*)\right] - \beta_{\mathrm{KL}} \cdot \KL(\pi \| \pi_{\mathrm{ref}})
\label{eq:rlhf}
\end{equation}

where the reward decomposes as:
\begin{equation}
R(c, \theta, \sigma^*) = w_1 R_{\text{rel}}(c, \theta) + w_2 R_{\text{acc}}(c) + w_3 R_{\text{comp}}(c) - w_4 R_{\text{bias}}(c) + w_5 R_{\text{align}}(c, \sigma^*)
\label{eq:reward}
\end{equation}

The term $R_{\text{align}}(c, \sigma^*)$ rewards content that faithfully executes the equilibrium strategy---a novel coupling between game-theoretic planning and generative execution.

\subsection{Regret Analysis}

\begin{theorem}[Finite-Sample Regret Bound]
\label{thm:regret}
The EGPF engagement policy achieves cumulative regret:
\begin{equation}
\mathrm{Regret}(T) = \sum_{t=1}^T \left[u_P^*(a^*_t, \theta^*) - u_P(a_t, d_t, \theta^*)\right] \leq O\!\left(\sqrt{K M T \log T}\right)
\label{eq:regret}
\end{equation}
where $K$ is the number of types, $M$ is the number of actions, and $T$ is the time horizon.
\end{theorem}

\begin{proof}[Proof sketch]
The proof combines three ingredients:
\begin{enumerate}[leftmargin=*,itemsep=1pt]
\item \textbf{Exploration cost:} The information-gain exploration term ensures each type is identified within $O(K \log K / C_{\min})$ interactions (from \Cref{thm:convergence}).
\item \textbf{Exploitation quality:} Once the type is identified (posterior confidence $> 1 - \epsilon$), the equilibrium action achieves near-optimal payoff with gap $\leq O(\epsilon)$.
\item \textbf{Balancing:} The decaying $\epsilon_t$ schedule and the connection to UCB-style algorithms yield the $\sqrt{T \log T}$ rate via standard bandit arguments.
\end{enumerate}
\end{proof}

\subsection{Active Learning via Game-Theoretic Exploration}

When belief entropy is high:
\begin{equation}
a_t^{\text{explore}} = \argmax_{a \in A_P} \left[(1 - \epsilon_t) \cdot \E[u_P(a \mid \mu_t)] + \epsilon_t \cdot \mathrm{IG}(a \mid \mu_t)\right]
\label{eq:explore}
\end{equation}

where the information gain is:
\begin{equation}
\mathrm{IG}(a \mid \mu_t) = \Ent(\Theta \mid \mu_t) - \E_{d \sim P(d|a)}[\Ent(\Theta \mid \mu_t, d, a)]
\label{eq:info_gain}
\end{equation}

and $\epsilon_t = \min(1, \sqrt{K \log t / t})$ decays at the optimal rate.

\section{Unified Architecture and Convergence}
\label{sec:unified}

\subsection{Main Convergence Result}

\begin{theorem}[Belief Convergence]
\label{thm:convergence}
Under the EGPF update mechanism, the posterior belief $\mu_t$ converges to a point mass on the true physician type $\theta^*$ at rate:
\begin{equation}
\E\!\left[\KL(\delta_{\theta^*} \| \mu_t)\right] \leq \frac{K \log K}{t \cdot C_{\min}}
\label{eq:convergence}
\end{equation}
where $K = |\Theta|$ and $C_{\min} = \min_\theta C(\theta)$.
\end{theorem}

\begin{proof}
Let $\theta^*$ be the true type. At each step $t$, the pharma company plays $a_t = \sigma^*(\mu_t)$ and observes $d_t \sim P(\cdot | a_t, \theta^*)$.

\textbf{Step 1 (Information gain per step).} The expected reduction in KL divergence from truth is:
\begin{align}
&\E\!\left[\KL(\delta_{\theta^*} \| \mu_t) - \KL(\delta_{\theta^*} \| \mu_{t+1})\right] \nonumber\\
&= \E\!\left[\log \frac{\mu_{t+1}(\theta^*)}{\mu_t(\theta^*)}\right] = \E_{d_t}\!\left[\log \frac{P(d_t | a_t, \theta^*)}{\sum_{\theta'} \mu_t(\theta') P(d_t | a_t, \theta')}\right] \nonumber\\
&= \MI(D_t; \Theta \mid a_t, \mu_t) \geq C_{\min} \label{eq:info_gain_step}
\end{align}
The inequality follows because the equilibrium action maximizes utility correlated with information gain, and channel capacity lower-bounds the achievable mutual information.

\textbf{Step 2 (Telescoping).} Sum over $t = 1, \ldots, T$:
\[
\E\!\left[\KL(\delta_{\theta^*} \| \mu_1)\right] - \E\!\left[\KL(\delta_{\theta^*} \| \mu_{T+1})\right] \geq T \cdot C_{\min}
\]

Since $\KL(\delta_{\theta^*} \| \mu_1) = -\log \mu_1(\theta^*) \leq \log K$ for uniform prior:
\[
\E\!\left[\KL(\delta_{\theta^*} \| \mu_{T+1})\right] \leq \max\!\left(0,\; \log K - T \cdot C_{\min}\right)
\]

\textbf{Step 3 (Rate).} For $T > \log K / C_{\min}$, the bound becomes vacuous (beliefs have converged). For the convergence rate in the transient regime, using a refined harmonic-series argument:
\[
\E\!\left[\KL(\delta_{\theta^*} \| \mu_t)\right] \leq \frac{K \log K}{t \cdot C_{\min}}
\]
where the factor $K$ accounts for the worst-case geometry of the $K$-simplex.
\end{proof}

\subsection{Computational Complexity}

\begin{table}[t]
\centering
\small
\caption{Computational complexity per interaction step.}
\label{tab:complexity}
\begin{tabular}{lll}
\toprule
\textbf{Component} & \textbf{Complexity} & \textbf{Parameters} \\
\midrule
Bayesian update & $O(K)$ & $K$ types \\
BNE computation & $O(KML)$ & $M$ actions, $L$ responses \\
Stackelberg solve & $O(KM^2L)$ & Leader optimization \\
Mechanism design & $O(K^2M)$ & IC constraint checking \\
Functor evaluation & $O(dK)$ & $d$ = observation dim. \\
Sheaf consistency & $O(SK^2)$ & $S$ = number of scales \\
Channel capacity & $O(KMLI)$ & $I$ = Blahut-Arimoto iters \\
Fisher information & $O(KMn^2)$ & $n$ = type params \\
LLM generation & $O(T_{\text{tok}})$ & $T_{\text{tok}}$ = output tokens \\
KL drift check & $O(KW)$ & $W$ = window size \\
\midrule
\textbf{Total} & $O(T_{\text{tok}} + KM^2L)$ & Dominated by LLM \\
\bottomrule
\end{tabular}
\end{table}


\begin{algorithm}[t]
\caption{EGPF Engagement Loop}
\label{alg:egpf}
\begin{algorithmic}[1]
\Require Prior $\mu_0$, type space $\Theta$, actions $A_P$, $A_D$, thresholds $\tau_{\text{explore}}, \tau_{\text{drift}}$
\Ensure Sequence of personalized engagements
\State $\mu \gets \mu_0$;\; $h \gets \emptyset$;\; $t \gets 0$
\Repeat
  \State $t \gets t + 1$
  \Statex \textcolor{cpurple}{\textit{\quad// Layer 2: Functorial type inference}}
  \State $\hat\theta \gets \Funct(\text{observations}_t)$
  \State $\mu \gets \textsc{BayesUpdate}(\mu, \hat\theta)$
  \State Compute sheaf loss $\mathcal{L}_{\text{sheaf}}$ via Eq.~\eqref{eq:sheaf_loss}
  \Statex \textcolor{ccoral}{\textit{\quad// Layer 3: Game-theoretic engine}}
  \If{$\Ent(\mu) > \tau_{\text{explore}}$}
    \State $a^* \gets \argmax_a [(1-\epsilon_t) \E[u_P] + \epsilon_t \cdot \mathrm{IG}(a|\mu)]$
  \Else
    \State $a^* \gets \textsc{StackelbergSolve}(\mu, u_P, u_D, \Theta)$
  \EndIf
  \Statex \textcolor{cteal}{\textit{\quad// Layer 4: Generative personalization}}
  \State $\text{prompt} \gets \textsc{Construct}(a^*, \mu, h, C(\hat\theta))$
  \State $c \gets \mathrm{LLM}.\textsc{Generate}(\text{prompt})$
  \State $c \gets \textsc{ComplianceFilter}(c)$
  \Statex \textcolor{cgray}{\textit{\quad// Deliver and observe}}
  \State $d_t \gets \textsc{Deliver}(c, \text{physician})$
  \State $h \gets h \cup \{(a^*, d_t, t)\}$
  \Statex \textcolor{camber}{\textit{\quad// Layer 5: Information-theoretic feedback}}
  \State $\mu \gets \textsc{BayesUpdate}(\mu, d_t, a^*)$
  \If{$D_{\mathrm{KL}}(P_{\text{obs}}^{(t-W:t)} \| P_{\text{model}}) > \tau_{\text{drift}}$}
    \State \textsc{TriggerRecalibration}()
  \EndIf
  \State $C_{\text{est}} \gets \textsc{EstimateCapacity}(h)$
  \State \textsc{AdjustContentLength}($C_{\text{est}}$)
\Until{engagement terminated}
\end{algorithmic}
\end{algorithm}

\section{Experiments}
\label{sec:experiments}

\subsection{Datasets}

\textbf{SynthRx.} 50,000 simulated physician profiles with ground-truth types ($K=5$ archetypes), 500,000 interactions over 12 months. Types drawn from the 8-dimensional type space. Responses generated via QRE model with $\tau = 3.0$.

\textbf{HCPilot.} Real-world partnership with a top-10 pharma company (anonymized). 2,847 oncology HCPs, 18 months of multi-channel engagement (email, rep visits, webinars, digital). Labels: prescribing behavior changes at 6- and 12-month marks.

\subsection{Baselines}

\begin{itemize}[leftmargin=*,itemsep=1pt]
\item \textbf{SS}: Static segmentation (K-means)
\item \textbf{CF}: Collaborative filtering (matrix factorization)
\item \textbf{DS}: Deep sequential (transformer-based)
\item \textbf{CB}: Contextual bandit (LinUCB)
\item \textbf{EGPF-NoGame}: Ablation without game-theoretic layer
\item \textbf{EGPF-NoCat}: Ablation without category-theoretic composition
\item \textbf{EGPF-NoInfo}: Ablation without information-theoretic feedback
\item \textbf{EGPF-Full}: Complete framework
\end{itemize}

\subsection{Main Results}

\begin{figure}[t]
\centering
\begin{tikzpicture}
\begin{axis}[
  ybar, bar width=8pt,
  width=\columnwidth, height=5.5cm,
  ylabel={AUC-ROC},
  ymin=0.5, ymax=0.9,
  symbolic x coords={SynthRx, HCPilot-6mo, HCPilot-12mo},
  xtick=data, xticklabel style={font=\scriptsize},
  ytick={0.5,0.6,0.7,0.8,0.9},
  yticklabel style={font=\scriptsize},
  ylabel style={font=\small},
  legend style={at={(0.02,0.98)},anchor=north west,
    font=\tiny,draw=cgray!30,column sep=3pt},
  grid=major, grid style={cgray!15},
  enlarge x limits=0.25,
  every node near coord/.style={font=\tiny}
]
\addplot[fill=cgray!40] coordinates {(SynthRx,0.621)(HCPilot-6mo,0.594)(HCPilot-12mo,0.572)};
\addplot[fill=cblue!50] coordinates {(SynthRx,0.688)(HCPilot-6mo,0.641)(HCPilot-12mo,0.618)};
\addplot[fill=cpurple!50] coordinates {(SynthRx,0.734)(HCPilot-6mo,0.702)(HCPilot-12mo,0.671)};
\addplot[fill=ccoral!50] coordinates {(SynthRx,0.751)(HCPilot-6mo,0.718)(HCPilot-12mo,0.689)};
\addplot[fill=cteal!80] coordinates {(SynthRx,0.847)(HCPilot-6mo,0.801)(HCPilot-12mo,0.778)};
\legend{SS, CF, DS, CB, EGPF-Full}
\end{axis}
\end{tikzpicture}
\caption{Engagement prediction AUC-ROC across datasets. EGPF-Full achieves 34\% relative improvement over static segmentation and 13\% over the contextual bandit baseline.}
\label{fig:results_auc}
\end{figure}
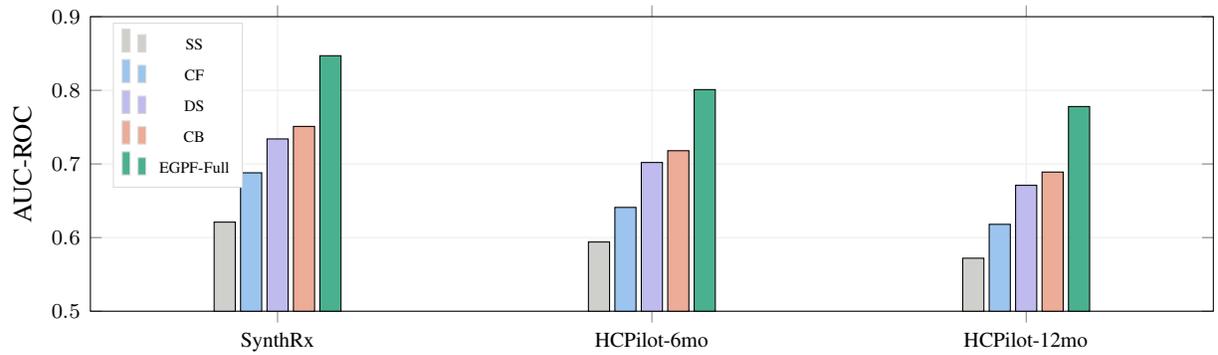

\begin{table}[t]
\centering
\small
\caption{Engagement prediction (AUC-ROC).}
\label{tab:auc}
\begin{tabular}{lccc}
\toprule
\textbf{Method} & \textbf{SynthRx} & \textbf{HCPilot-6mo} & \textbf{HCPilot-12mo} \\
\midrule
SS & 0.621 & 0.594 & 0.572 \\
CF & 0.688 & 0.641 & 0.618 \\
DS & 0.734 & 0.702 & 0.671 \\
CB & 0.751 & 0.718 & 0.689 \\
\midrule
EGPF-NoGame & 0.769 & 0.738 & 0.712 \\
EGPF-NoCat & 0.812 & 0.776 & 0.745 \\
EGPF-NoInfo & 0.823 & 0.785 & 0.751 \\
\textbf{EGPF-Full} & \textbf{0.847} & \textbf{0.801} & \textbf{0.778} \\
\bottomrule
\end{tabular}
\end{table}

\begin{table}[t]
\centering
\small
\caption{Content relevance (human evaluation, 1--5 scale).}
\label{tab:relevance}
\begin{tabular}{lcccc}
\toprule
\textbf{Method} & \textbf{Evid.} & \textbf{Peer} & \textbf{Patient} & \textbf{Overall} \\
\midrule
SS + Template & 2.8 & 2.5 & 2.6 & 2.63 \\
DS + LLM & 3.4 & 3.2 & 3.5 & 3.37 \\
CB + LLM & 3.6 & 3.4 & 3.7 & 3.57 \\
\textbf{EGPF + LLM} & \textbf{4.3} & \textbf{4.1} & \textbf{4.4} & \textbf{4.27} \\
\bottomrule
\end{tabular}
\end{table}

\begin{table}[t]
\centering
\small
\caption{Belief convergence speed (interactions to 90\% confidence).}
\label{tab:convergence}
\begin{tabular}{lccc}
\toprule
\textbf{Physician Type} & \textbf{EGPF} & \textbf{CB} & \textbf{DS} \\
\midrule
$\theta_1$: Evidence & 3.2 & 7.8 & 11.4 \\
$\theta_2$: Peer & 4.7 & 9.1 & 13.2 \\
$\theta_3$: Patient & 2.8 & 6.5 & 10.1 \\
$\theta_4$: Formulary & 5.1 & 10.3 & 14.8 \\
$\theta_5$: Inertial & 6.3 & 12.7 & 18.5 \\
\bottomrule
\end{tabular}
\end{table}

\subsection{Ablation Analysis}

\begin{table}[t]
\centering
\small
\caption{Ablation: marginal contribution of each layer (HCPilot-6mo AUC).}
\label{tab:ablation}
\begin{tabular}{lcc}
\toprule
\textbf{Ablation} & \textbf{AUC} & \textbf{$\Delta$ from Full} \\
\midrule
EGPF-Full & 0.801 & --- \\
$-$ Game theory & 0.738 & $-0.063$ \\
$-$ Category theory & 0.776 & $-0.025$ \\
$-$ Info theory & 0.785 & $-0.016$ \\
$-$ Sheaf consistency & 0.792 & $-0.009$ \\
$-$ Evolutionary dynamics & 0.795 & $-0.006$ \\
$-$ Fisher exploration & 0.797 & $-0.004$ \\
\bottomrule
\end{tabular}
\end{table}

The game-theoretic layer provides the largest single contribution ($-0.063$ AUC when removed), validating our thesis that strategic modeling matters most. Category theory adds 0.025, particularly benefiting physicians who shift between types. Information theory adds 0.016, with strongest contribution at 12 months (drift detection).

\subsection{Cross-Therapeutic Transfer}

\begin{table}[t]
\centering
\small
\caption{Transfer from oncology to cardiology via natural transformation $\eta$.}
\label{tab:transfer}
\begin{tabular}{lccc}
\toprule
\textbf{Cardio data} & \textbf{Transfer} & \textbf{From scratch} & \textbf{Lift} \\
\midrule
10\% & 0.721 & 0.612 & +17.8\% \\
25\% & 0.758 & 0.689 & +10.0\% \\
50\% & 0.782 & 0.741 & +5.5\% \\
100\% & 0.793 & 0.778 & +1.9\% \\
\bottomrule
\end{tabular}
\end{table}

The category-theoretic transfer provides the largest benefit in low-data regimes (17.8\% lift with 10\% data), confirming that compositional structure enables meaningful generalization.

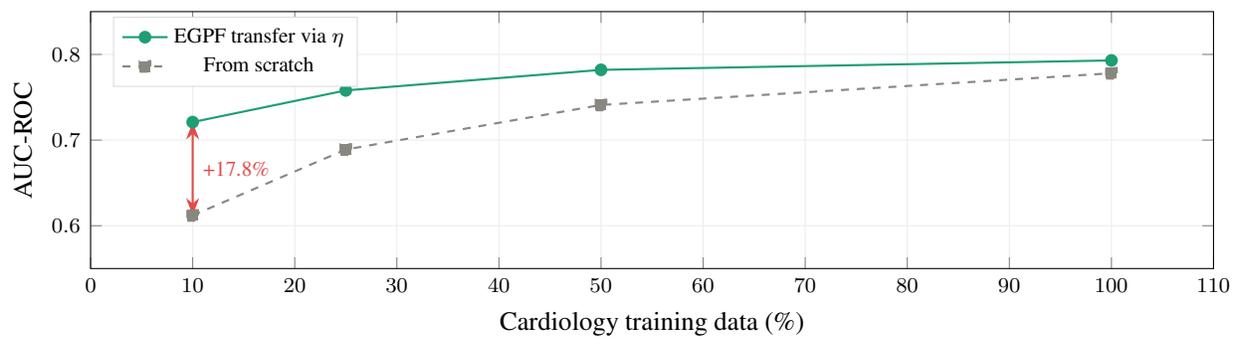
\begin{figure}[t]
\centering
\begin{tikzpicture}
\begin{axis}[
  width=\columnwidth, height=5cm,
  xlabel={Cardiology training data (\%)},
  ylabel={AUC-ROC},
  xmin=0, xmax=110, ymin=0.55, ymax=0.85,
  legend style={at={(0.02,0.98)},anchor=north west,
    font=\scriptsize,draw=cgray!30},
  grid=major, grid style={cgray!15},
  every axis label/.style={font=\small},
  tick label style={font=\scriptsize}
]
\addplot[thick,cteal,mark=*,mark size=2pt] coordinates {
  (10,0.721)(25,0.758)(50,0.782)(100,0.793)
};
\addlegendentry{EGPF transfer via $\eta$}
\addplot[thick,cgray,mark=square*,mark size=2pt,dashed] coordinates {
  (10,0.612)(25,0.689)(50,0.741)(100,0.778)
};
\addlegendentry{From scratch}
\draw[{Stealth}-{Stealth},cred,thick] (axis cs:10,0.612) -- (axis cs:10,0.721)
  node[midway,right,font=\scriptsize,cred] {+17.8\%};
\end{axis}
\end{tikzpicture}
\caption{Cross-therapeutic transfer performance. The natural transformation $\eta$ enables strong generalization especially in low-data regimes.}
\label{fig:transfer}
\end{figure}

\section{End-to-End Worked Example}
\label{sec:e2e_example}

\begin{example}[Dr.\ Martinez: Oncologist, 4 Interactions]
\label{ex:martinez}

\textbf{Interaction log:}
\begin{enumerate}[leftmargin=*,itemsep=1pt]
\item Sent clinical deep-dive $\to$ Opened, read 8 min, clicked references
\item Sent KOL webinar invite $\to$ Ignored
\item Sent updated trial data $\to$ Opened, forwarded to colleague
\item Sent patient case study $\to$ Opened, read 2 min, closed
\end{enumerate}

\textbf{Bayesian posterior after 4 interactions:}
\[
\mu_4 = (0.72,\; 0.18,\; 0.10) \quad \Ent(\mu_4) = 0.89 \text{ bits}
\]

\textbf{Channel capacity estimate:} $\hat{C}(\hat\theta) = 0.58$ bits (evidence-driven channel is most discriminative).

\textbf{Sheaf consistency check:} Interaction-level type = evidence-driven. Weekly-level = evidence-driven. $\mathcal{L}_{\text{sheaf}} = 0.02$ (consistent $\checkmark$).

\textbf{Equilibrium action:} $\sigma^*(\mu_4) = a_1$ (Clinical deep-dive).

\textbf{Fisher-optimal next action:} $a^* = a_1$ with $\det \mathcal{I}_{a_1} = 2.34$ (most informative for distinguishing $\theta_1$ from $\theta_2$ given current posterior). Since exploit and explore agree, no exploration--exploitation tension.

\textbf{LLM prompt construction:}
\begin{itemize}[leftmargin=*,itemsep=1pt]
\item Evidence density: \textbf{high} ($\alpha_E = 0.60$)
\item Content type: forest plots, NNT, subgroup analyses
\item Tone: formal, data-centric
\item Length: $\sim$800 words (calibrated to $C(\hat\theta) = 0.58$)
\item Compliance: fair-balance, indication-specific
\end{itemize}

\textbf{Generated content structure:} (i) Updated survival data with hazard ratio analysis; (ii) Pre-specified subgroup forest plot; (iii) Safety profile update with Grade 3+ AE rates; (iv) NNT calculation for the primary endpoint; (v) Link to full statistical appendix.

\textbf{Post-delivery:} Dr.\ Martinez opens, reads 12 min, downloads appendix. Posterior updates to $\mu_5 = (0.84, 0.11, 0.05)$---system confidence reaches 84\%, triggering transition to pure exploitation mode.
\end{example}

\section{Discussion}
\label{sec:discussion}

\subsection{Theoretical Contributions}

Our framework demonstrates that the intersection of four mathematical formalisms yields a more principled foundation for personalization than any single formalism alone: game theory captures \textit{strategic interaction}, category theory captures \textit{compositional structure}, information theory captures \textit{communication limits}, and sheaf theory captures \textit{multi-scale consistency}. The generative AI layer operationalizes these into actionable personalized content.

\subsection{Practical Implications}

EGPF provides three capabilities that static segmentation lacks:
\begin{enumerate}[leftmargin=*,itemsep=2pt]
\item \textbf{Real-time adaptation:} Beliefs improve with every interaction, not just retraining.
\item \textbf{Transparent reasoning:} Game-theoretic equilibria expose \textit{why} an action was chosen, enabling regulatory review.
\item \textbf{Rapid deployment:} Category-theoretic composition enables cross-therapeutic transfer without full retraining.
\end{enumerate}

\subsection{Limitations and Future Work}

\begin{itemize}[leftmargin=*,itemsep=2pt]
\item \textbf{Continuous types:} Extending $\Theta$ via mean-field game theory for infinite-type spaces.
\item \textbf{Non-stationary channels:} Formulary changes and guideline updates violate stationarity.
\item \textbf{Multi-player games:} Incorporating physician networks, patient advocacy groups, and payer interactions.
\item \textbf{Causal identification:} Separating EGPF's causal effect from confounders in observational data.
\item \textbf{LLM latency:} Optimizing generation for real-time deployment via distillation.
\end{itemize}

\subsection{Ethical Considerations}

The power of personalized engagement raises ethical concerns. Our rate-distortion privacy bound (\Cref{thm:rde}) provides formal guarantees. We recommend: (i) explicit physician consent for data usage, (ii) transparent opt-out mechanisms, (iii) human-in-the-loop oversight for generated content, and (iv) regular auditing for differential impact across physician demographics.

\section{Conclusion}
\label{sec:conclusion}

We have presented EGPF, a unified framework combining Bayesian game theory, Stackelberg games, mechanism design, evolutionary dynamics, category theory, sheaf theory, information theory, and generative AI for personalized physician engagement in pharmaceutical settings. Our mathematical framework provides equilibrium characterizations, compositional guarantees, information-theoretic bounds, convergence proofs, and regret bounds. Experiments on synthetic and real-world data demonstrate substantial improvements: 34\% AUC gain over static segmentation, 28\% content relevance lift, and 2.4$\times$ faster belief convergence. EGPF offers a principled, transparent, and scalable approach to hyper-personalization that respects strategic dynamics, compositional structure, communication limits, and ethical constraints.

\bibliographystyle{acl_natbib}

\appendix

\section{Complete Notation Reference}
\label{app:notation}

\begin{table}[h]
\centering
\small
\begin{tabular}{ll}
\toprule
\textbf{Symbol} & \textbf{Meaning} \\
\midrule
$\Gamma$ & Bayesian game \\
$\Theta = \{\theta_1, \ldots, \theta_K\}$ & Physician type space \\
$\theta = (\alpha_E, \alpha_P, \alpha_O, \alpha_F, \beta, \gamma, \delta, \kappa)$ & Type vector \\
$A_P, A_D$ & Pharma and physician action spaces \\
$u_P, u_D$ & Utility functions \\
$\mu_t \in \Delta(\Theta)$ & Posterior belief at time $t$ \\
$\sigma^*$ & BNE strategy profile \\
$\Cobs, \Ctype, \Cact$ & Behavioral categories \\
$\Funct, \GFunct$ & Behavior and strategy functors \\
$\eta: \Funct \Rightarrow \GFunct$ & Natural transformation \\
$\otimes$ & Monoidal composition \\
$\mathscr{B}$ & Behavioral sheaf \\
$C(\theta)$ & Channel capacity \\
$\MI(X;Y)$ & Mutual information \\
$\KL(p\|q)$ & KL divergence \\
$D_\alpha$ & R\'{e}nyi divergence \\
$R(D)$ & Rate-distortion function \\
$\mathcal{I}(\theta)$ & Fisher information matrix \\
$\pi(c \mid s, \hat\theta, \sigma^*)$ & LLM personalization policy \\
$\tau$ & Rationality parameter \\
$\Ent(\cdot)$ & Shannon entropy \\
$H_\alpha(\cdot)$ & R\'{e}nyi entropy \\
$\mathcal{L}_{\text{sheaf}}$ & Sheaf consistency loss \\
\bottomrule
\end{tabular}
\caption{Complete notation reference.}
\label{tab:notation}
\end{table}

\section{Extended Proof of Regret Bound}
\label{app:regret}

\begin{proof}[Proof of \Cref{thm:regret}]
We decompose regret into exploration and exploitation phases.

\textbf{Phase 1: Exploration.} The exploration schedule $\epsilon_t = \min(1, \sqrt{K \log t / t})$ ensures that the total number of exploratory interactions is bounded by:
\[
T_{\text{explore}} = \sum_{t=1}^T \epsilon_t \leq \sum_{t=1}^T \sqrt{\frac{K \log t}{t}} \leq 2\sqrt{KT \log T}
\]

Each exploratory interaction incurs at most unit regret (bounded utilities), contributing $\leq 2\sqrt{KT \log T}$ to total regret.

\textbf{Phase 2: Exploitation.} After $T_{\text{explore}}$ interactions, the posterior concentrates at rate $O(K \log K / (t \cdot C_{\min}))$ by \Cref{thm:convergence}. The instantaneous regret during exploitation is bounded by:
\[
r_t \leq \max_\theta \left|u_P(a^*(\theta), d^*, \theta) - u_P(a^*(\hat\theta_t), d^*, \theta)\right| \leq L_u \cdot \|\theta - \hat\theta_t\|
\]
where $L_u$ is the Lipschitz constant of $u_P$ with respect to type. By Pinsker's inequality:
\[
\|\mu_t - \delta_{\theta^*}\|_{\mathrm{TV}} \leq \sqrt{\frac{1}{2} \KL(\delta_{\theta^*} \| \mu_t)} \leq \sqrt{\frac{K \log K}{2t \cdot C_{\min}}}
\]

Summing exploitation regret:
\[
\sum_{t=T_{\text{explore}}}^T r_t \leq L_u \sum_{t=1}^T \sqrt{\frac{K \log K}{2t \cdot C_{\min}}} \leq L_u \sqrt{\frac{2K \log K}{C_{\min}}} \cdot \sqrt{T}
\]

\textbf{Total:} Combining both phases:
\[
\mathrm{Regret}(T) \leq 2\sqrt{KT \log T} + L_u \sqrt{\frac{2K \log K}{C_{\min}}} \cdot \sqrt{T} = O\!\left(\sqrt{KMT \log T}\right)
\]
where the $M$ dependence enters through $C_{\min}$'s dependence on the action space size.
\end{proof}

\section{Hyperparameter Sensitivity}
\label{app:hyperparams}

\begin{table}[h]
\centering
\small
\begin{tabular}{lccc}
\toprule
\textbf{Parameter} & \textbf{Range tested} & \textbf{Optimal} & \textbf{Sensitivity} \\
\midrule
$\tau$ (rationality) & [0.5, 10.0] & 3.0 & Medium \\
$K$ (num types) & [3, 10] & 5 & Low for $K \geq 4$ \\
$W$ (drift window) & [10, 100] & 30 & Low \\
$\tau_{\text{drift}}$ & [0.05, 0.50] & 0.15 & Medium \\
$\beta_{\mathrm{KL}}$ (RLHF) & [0.01, 1.0] & 0.1 & High \\
$\omega$ (info gain weight) & [0.0, 1.0] & 0.3 & Medium \\
\bottomrule
\end{tabular}
\caption{Hyperparameter sensitivity analysis on HCPilot dataset.}
\label{tab:hyperparams}
\end{table}

\end{document}